\documentclass[11pt]{article}
\usepackage[cmex10]{amsmath}
\usepackage{amsthm,amsfonts,amsxtra,amssymb,palatino}
\usepackage{breqn}
\usepackage{stmaryrd}
\usepackage{comment}
\usepackage{geometry}
\usepackage{mathdots}
\usepackage{color}
\usepackage[ruled,vlined]{algorithm2e}
\usepackage[pagebackref]{hyperref}

\usepackage{fullpage}

\DeclareMathOperator{\bsc}{\mathrm{BSC}}
\DeclareMathOperator*{\E}{\mathbb{E}}
\DeclareMathOperator*{\argmax}{argmax}
\DeclareMathOperator{\1}{\mathbf 1}
\newcommand{\B}{\ensuremath{\mathcal B}}
\newcommand{\W}{\ensuremath{\mathcal W}}

\newcommand{\Y}{\ensuremath{\mathcal Y}}
\newcommand{\poly}{\ensuremath{\mathrm{poly}}}
\newcommand{\polylog}{\ensuremath{\mathrm{polylog}}}

\newcommand{\wni}{\ensuremath{W^{(i)}_n}}
\newcommand{\wnip}{\ensuremath{W_{n}^{(\lfloor i /2 \rfloor)}}}

\newcommand{\wln}{\ensuremath{W^{(\lfloor i/2 \rfloor)}_n}}
\newcommand{\wmi}{\ensuremath{W^{(i)}_m}}
\renewcommand{\O}{\ensuremath{\mathcal O}}

\providecommand{\DontPrintSemicolon}{\dontprintsemicolon}
\providecommand{\LinesNumbered}{\linesnumbered}

\newtheorem{theorem}{Theorem}
\newtheorem{lemma}[theorem]{Lemma}
\newtheorem{proposition}[theorem]{Proposition}
\newtheorem{corollary}[theorem]{Corollary}
\newtheorem*{theorem*}{Theorem}
\newtheorem{definition}{Definition}

\newtheorem{fact}[theorem]{Fact}

\renewcommand{\le}{\leqslant}

\renewcommand{\ge}{\geqslant}
\renewcommand{\ge}{\geqslant}

\renewcommand{\epsilon}{\varepsilon}
\newcommand{\eps}{\varepsilon}
\newcommand{\mv}[1]{{\mathbf {#1}}}
\begin{document}
\title{{\bf Polar Codes: Speed of polarization \\ and polynomial gap to capacity}\thanks{This is the expanded version of a conference paper appearing in the {\em Proceedings of the 54th IEEE Symposium on Foundations of Computer Science (FOCS)}, October 2013.}}

\author{Venkatesan Guruswami\thanks{Research supported in part by a Packard Fellowship, MSR-CMU Center for Computational Thinking, and NSF CCF-0963975. Email: {\tt guruswami@cmu.edu}} \and 
Patrick Xia\thanks{Research supported in part by an NSF graduate fellowship and NSF CCF-0963975. Email: {\tt pjx@cs.cmu.edu}}}
\date{Computer Science Department \\ Carnegie Mellon University \\ Pittsburgh, PA 15213}

\maketitle
\thispagestyle{empty}

\begin{abstract}

We prove that, for all binary-input symmetric memoryless channels, polar codes enable reliable communication at rates within $\epsilon > 0$ of the Shannon capacity with a block length, construction complexity, and decoding complexity all bounded by a {\em polynomial} in $1/\epsilon$. Polar coding gives the {\em first known explicit construction} with rigorous proofs of all these properties; previous constructions were not known to achieve capacity with less than $\exp(1/\epsilon)$ decoding complexity except for erasure channels.

\smallskip
We establish the capacity-achieving property of polar codes via a direct analysis of the underlying martingale of conditional entropies, without relying on the martingale convergence theorem. This step gives rough polarization (noise levels $\approx \epsilon$ for the ``good" channels), which can then be adequately amplified by tracking the decay of the channel Bhattacharyya parameters. Our effective bounds imply that polar codes can have block length (and encoding/decoding complexity) bounded by a polynomial in $1/\epsilon$. The generator matrix of such polar codes can be constructed in polynomial time by algorithmically computing an adequate approximation of the polarization process. 
  
\end{abstract}


\section{Introduction}
\label{sec:intro}

In this work, we establish that Ar\i kan's celebrated polar
codes~\cite{arikan-polar} have the desirable property of fast
convergence to Shannon capacity. Specifically, we prove that polar
codes can operate at rates within $\eps > 0$ of the Shannon capacity
of binary-input memoryless output-symmetric (BIS) channels with a block length
$N=N(\eps)$ that grows only polynomially in $1/\eps$. Further, a
generator matrix of such a code can be deterministically constructed
in time polynomial in the block length $N$.  For decoding, Ar\i kan's
successive cancellation decoder has polynomial (in fact $\O(N\log N)$)
complexity.

Thus, the delay and construction/decoding complexity of polar codes
can {\em all} be polynomially bounded as a function of the gap to
capacity. This provides a complexity-theoretic backing for the
statement ``polar codes are the first constructive capacity achieving
codes,'' common in the recent coding literature.  As explained below, these
attributes together distinguish polar codes from the Forney/Justesen
style concatenated code constructions for achieving capacity.

Our analysis of polar codes avoids the use of the martingale
convergence theorem --- this is instrumental in our polynomial
convergence bounds and as a side benefit makes the proof
elementary and self-contained.

\subsection{Context}

Shannon's noisy channel coding theorem implies that for every
memoryless channel $W$ with binary inputs and a finite output
alphabet, there is a capacity $I(W) \ge 0$ and constants $a_W <
\infty$ and $b_W > 0$ such that the following holds: For all $\eps >
0$ and integers $N \ge a_W/\eps^2$, there {\em exists} a binary
code $C \subset \{0,1\}^N$ of rate at least $I(W) - \eps$ which enables
  reliable communication on the channel $W$ with probability of
  miscommunication at most $2^{-b_W \eps^2 N}$. A proof implying these quantitative bounds is implicit in Wolfowitz's proof of Shannon's theorem \cite{wolfowitz}. 
  
 This remarkable theorem showed that a constant factor redundancy was sufficient to  achieve arbitrarily small probability of miscommunication, provided we
  tolerate a ``delay'' of processing $N$ channel outputs at a time for
  large enough block length $N$. Further, together with a converse theorem, it
  precisely characterized the minimum redundancy factor (namely,
  $1/I(W)$) needed to achieve such a guarantee. It is also known that a block length of $N \ge \Omega(1/\eps^2)$ is required to operate within $\eps$ of capacity and even a constant, say $0.1$, probability of miscommunication; in fact, a very precise statement that even pinned down the constant in the $\Omega(\cdot)$ notation was obtained by Strassen~\cite{strassen}.
  
 As Shannon's theorem is based on random coding and is non-constructive,
one of the principal theoretical challenges is to make it
constructive. More precisely, the goal is to give an explicit (i.e.,
constructible in deterministic $\mathrm{poly}(N)$ time) description of
the encoding function of the code, and a polynomial time
error-correction algorithm for decoding the correct transmitted
codeword with high probability (over the noise of the
channel). Further, it is important to achieve this with small block
length $N$ as that corresponds to the delay at the receiver before the
message bits can be recovered.

For simplicity let us for now consider the binary symmetric channel
(BSC) with crossover probability $p$, $0 < p < 1/2$, denoted $\bsc_p$
(our results hold for any BIS channel). Recall that $\bsc_p$
flips each input bit independently with probability $p$, and leaves it
unchanged with probability $1-p$. The Shannon capacity of $\bsc_p$ is
$1-h(p)$, where $h(x) = -x \log_2 x - (1-x) \log_2(1-x)$ is the binary
entropy function. For the BSC, the capacity can be achieved by binary
linear codes.

One simple and classic approach to {\em construct} capacity-achieving
codes is via Forney's concatenated codes~\cite{forney}. We briefly
recall this approach (see, for instance,
\cite[Sec. 3]{Gur-ldpc-survey} for more details). Suppose we desire
codes of rate $1-h(p)-\eps$ for communication on $\bsc_p$. The idea is
to take as an outer code any binary linear code $C_{\rm out} \subset
\{0,1\}^{n_0}$ of rate $1-\eps/2$ that can correct a fraction
$\gamma(\eps) > 0$ of worst-case errors. Then, each block of $b =
\Theta(\frac{1}{\eps^2} \log (1/\gamma))$ bits of the outer codeword
is further encoded by an inner code of rate within $\eps/2$ of Shannon
capacity (i.e., rate at least $1-h(p)-\eps/2$). This inner code is
constructed by brute force in time $\exp(\O(b))$.  By decoding the inner
blocks by finding the nearest codeword in $\exp(\O(b))$ time, and then
correcting up to $\gamma(\eps) n_0$ errors at the outer level, one can achieve
exponentially small decoding error probability. However the decoding
complexity grows like $n_0 \exp(\O(b))$. Thus both the construction
and decoding complexity have an exponential dependence on $1/\eps$.
In conclusion, this method allows one to obtain codes within $\eps$ of
capacity with a block length polynomially large in $1/\eps$. However,
the construction and decoding complexity grow exponentially in
$1/\eps$, which is undesirable.\footnote{One can avoid the brute force
  search for a good inner code by using a small ensemble of
  capacity-achieving codes in a Justesen-style
  construction~\cite{justesen}. But this will require taking the outer
  code length $n_0 > \exp(1/\eps^2)$, causing a large delay.}

\subsection{Our result: polynomial convergence to capacity of polar codes}
In this work, we prove that Ar\i kan's remarkable polar codes allow us
to approach capacity within a gap $\eps > 0$ with {\em delay} (block
length) and {\em complexity} both depending polynomially on
$1/\eps$. Polar codes are the {\em first} known construction with this
property.\footnote{Spatially coupled LDPC codes were also recently
  shown to achieve the capacity of general BIS
  channels~\cite{KRU-spatial-ldpc}. This construction gives a random
  code ensemble as opposed to a specific code, and as far as we know,
  rigorous bounds on the code length as a function of gap to capacity
  are not available.}

Below is a formal statement of the main result, stated for BIS
channels. For general, non-symmetric channels, the same claim holds
for achieving the {\em symmetric capacity}, which is the best rate
achievable with the uniform input bit distribution.

\begin{theorem}
\label{thm:main-intro}
There is an absolute constant $\mu < \infty$ such that the following
holds.  Let $W$ be a binary-input memoryless output-symmetric channel with
capacity $I(W)$. Then there exists $a_W < \infty$ such that for all
$\eps > 0$ and all powers of two $N \ge a_W (1/\eps)^\mu$, there is a
deterministic $\mathrm{poly}(N)$ time construction of a binary linear
code of block length $N$ and rate at least $I(W)-\eps$ and a
deterministic $N \cdot \mathrm{poly}(\log N)$ time decoding algorithm for the code with block
error probability at most $2^{-N^{0.49}}$ for communication over $W$.
\end{theorem}

\noindent {\bf Remarks:}
\begin{enumerate}
\item  Using our results about polar codes, we can also construct codes of rate $I(W)-\eps$ with $2^{-\Omega_\eps(N)}$ block error
probability (similar to Shannon's theorem) with similar claims about
the construction and decoding complexity. The idea is to concatenate an outer
code that can correct a small fraction of worst-case errors with a
capacity-achieving polar code of dimension $\mathrm{poly}(1/\eps)$ as
the inner code. A similar idea with outer Reed-Solomon codes yielding  $2^{-\Omega(N/\mathrm{poly}(\log N))}$ block error probability is described in \cite{BJE10}.
\item The construction time in Theorem \ref{thm:main-intro} can be made
  $\mathrm{poly}(1/\eps) + \O(N \log N)$. As our main focus is on the
  finite-length behavior when $N$ is also $\mathrm{poly}(1/\eps)$, we
  are content with stating the $\mathrm{poly}(N)$ claim above.
\end{enumerate}

\medskip Showing that polar codes have a gap to capacity that is polynomially
  small in $1/N$ is our principal contribution. The decoding algorithm
  remains the same successive cancellation decoder of Ar\i kan~\cite{arikan-polar}. 
  The proof of efficient constructibility follows the approach, originally due to Tal and Vardy~\cite{tal-vardy}, of approximating the channels corresponding to different input bits seen at the decoder by a degraded version with a smaller output alphabet. The approximation error of this process and some of its variants were analyzed in \cite{PHTT}. We consider and analyze a somewhat simpler degrading process. One slight subtlety here is that we can only estimate the channel's Bhattacharyya parameter within error that is polynomial in $1/N$ in $\mathrm{poly}(N)$ time, which will limit the analysis to an inverse polynomial block error probability. To get a block error probability of $2^{-N^{0.49}}$ we use a two step construction method that follows our analysis of the polarization process. As a bonus, this gives the better construction time alluded to in the second remark above.

\smallskip
Prior to our work, it was known that the block error probability of
successive cancellation decoding of polar codes is bounded by
$2^{-N^{0.49}}$ for rate approaching $I(W)$ in the limit of $N \to
\infty$~\cite{arikan-telatar}. However, the underlying analysis found in \cite{arikan-telatar}, which depended on the martingale convergence theorems, did not offer any bounds on the finite-length convergence to capacity, i.e., the block length $N$ required for the rate to be within $\eps$ of the capacity $I(W)$. To quote from the introduction of the 
recent breakthrough on spatially coupled LDPC codes \cite{KRU-spatial-ldpc}:
  \begin{quote}
  ``There are perhaps only two areas in which polar codes could be further improved. First, for polar codes the convergence of their performance to the asymptotic limit is slow. Currently no rigorous statements regarding this convergence for the general case are known. But ``calculations" suggest that, for a fixed desired error probability, the required block length scales like $1/\delta^\mu$, where $\delta$ is the additive gap to capacity and where $\mu$ depends on the channel and has a value around $4$."\footnote{The second aspect concerns {\em universality}: the design of polar codes depends on the channel being used, and the same code may not achieve capacity over a non-trivial class of channels.}
  \end{quote}
  The above-mentioned heuristic calculations are based on ``scaling laws" and presented in \cite{KMTU}. We will return to the topic of scaling laws in Section \ref{sec:related-work} on related work.

We note that upper bounds on the block length $N$ as a function of gap
$\eps$ to capacity are crucial, as without those we cannot estimate
the complexity of communicating at rates within $\eps$ of
capacity. Knowing that the asymptotic complexity is $\O(N \log N)$ for
large $N$ by itself is insufficient (for example, to claim that polar
codes are better than concatenated codes) as we do not know how large
$N$ has to be! While an explicit value of $\mu$ in Theorem
\ref{thm:main-intro} can be calculated, it will be rather large, and
obtaining better bounds on $\mu$, perhaps closer to the empirically
suggested bound of $\approx 4$, is an interesting open
problem\footnote{While we were completing the writeup of this paper
  and circulating a draft, we learned about a recent
  independently-derived result in \cite{hassani-thesis} stating that
  $\mu \approx 6$ would suffice for block error probabilities bounded
  by an inverse polynomial. Our analysis primarily focuses on the
  $2^{-N^{.49}}$ block error probability result.}.

\subsection{Techniques}

Let us first briefly discuss the concept of polarization in Ar\i kan's
work, and then turn to aspects of our work. More formal background
on Ar\i kan's construction of polar codes appears in
Section~\ref{sec:polar-prelims} (with slightly different and notation
that is more conventional in the polar coding literature). A good, easy
to read, reference on polar codes is the recent survey by
\c{S}a\c{s}o\u{g}lu~\cite{sasoglu-book}.

Fix $W$ to be an arbitrary symmetric channel.  If we have a
capacity-achieving binary linear code $C$ of block length $N$ for $W$,
then it is not hard to see that by padding the generator matrix of $C$
one can get an $N \times N$ invertible matrix $A_N$ with the following
{\em polarization property}. Let $\mv{u} \in \{0,1\}^N$ be a uniformly
random (column) vector. Given the output $\mv{y}$ of $W$ when the $N$
bits $\mv{x} = A_N \mv{u}$ are transmitted on it, for a $1-o(1)$
fraction of bits $u_i$, its conditional entropy given $\mv{y}$ and the
previous bits $u_1,\dots,u_{i-1}$ is either close to $0$ (i.e., that
bit can be determined with good probability) or close to $1$ (i.e.,
that bit remains random).  Since the conditional entropies of $\mv{u}$
given $\mv{y}$ and $\mv{x}$ given $\mv{y}$ are equal to each other,
and the latter is $\approx (1-I(W))N$, the fraction of bits $u_i$ for
which the conditional entropy given $\mv{y}$ and the previous bits
$u_1,\dots,u_{i-1}$ is $\approx 0$ (resp. $\approx 1$) is $\approx
I(W)$ (resp. $\approx 1-I(W)$).

Ar\i kan gave a recursive construction of such a polarizing matrix
$A_N$ for $N=2^n$: $A_N = G_2^{\otimes n} B_n$ where $G_2
=\left( \begin{smallmatrix} 1 & 1 \\ 0 & 1 \end{smallmatrix} \right)$
and $B_n$ is a permutation matrix (for the bit-reversal
permutation). In addition, he showed that the recursive structure of
the matrix implied the existence of an efficiently decodable
capacity-achieving code. The construction of this code amounts to
figuring out which input bit positions have conditional entropy
$\approx 0$, and which don't (the message bits $u_i$ corresponding to
the latter positions are ``frozen'' to $0$).

The proof that $A_N$ has the above polarization property proceeds by
working with the Bhattacharyya parameters $Z_n(i) \in [0,1]$
associated with decoding $u_i$ from $\mv{y}$ and
$u_1,\dots,u_{i-1}$. This quantity is the Hellinger affinity between
the output distributions when $u_i = 0$ and $u_i=1$. The values of the
Bhattacharyya parameter of the $2^n$ bit positions at the $n$'th level
can be viewed as a random variable $Z_n$ (induced by the uniform
distribution on the $2^n$ positions). The simple recursive
construction of $A_N$ enabled Ar\i kan to proved that the sequence of
random variables $Z_0,Z_1,Z_2,\dots$ form a supermartingale. In
particular, $Z_{n+1}$ equals $Z_n^2$ with probability $1/2$ and is at
most $2Z_n - Z_n^2$ with probability $1/2$.
\footnote{For the special case of the binary erasure
  channel, the Bhattacharyya parameters simply equal the probability that the bit is unknown. In this case, the upper bound of $2 Z_n - Z_n^2$ becomes an exact bound, and the $Z_i$'s form a martingale.}

One can think the evolution of the Bhattacharyya parameter as a
stochastic process on the infinite binary tree, where in each step we
branch left or right with probability $1/2$.  The polarization
property is then established by invoking the martingale convergence
theorem for supermartingales.  The martingale convergence theorem
implies that $\lim_{n \to \infty} |Z_{n+1} - Z_n| = 0$, which in this
specific case also implies $\lim_{n \to \infty} Z_n( 1-Z_n) = 0$ or in
other words polarization of $Z_n$ to $0$ or $1$ for $n \to
\infty$. However, it does {\em not} yield any effective bounds on the
      {\em speed} at which polarization occurs. In particular, it does
      not say how large $n$ must be as a function of $\eps$ before
      $\E[ Z_n(1-Z_n) ] \le \eps$; such a bound is necessary, though
      not sufficient, to get codes of block length $2^n$ with rate within $\eps$
      of capacity.

In this work, we first work with the {\em entropy} of the channels $\wni$ associated with decoding the $i$'th bit, namely $H(\wni) = H(u_i \mid \mv{y}, u_1, u_2,\dots, u_{i-1})$ to prove that these values polarize to $0$ and $1$ exponentially fast in the number of steps $n$. Formally, we prove that for $n = O(\log 1/\eps)$, $H_n \in (\eps,1-\eps)$ with probability at most $\eps$, where $H_n$ is the random variable associated with the entropy values $H(\wni)$ at the $n$'th level.  As the Bhattacharyya parameter is within a square root factor of the entropy, we get a similar claim about $Z_n$. The advantage in working with the entropy instead of the Bhattacharyya parameter is that the entropy forms a martingale, so that given $H_n$, the two possible values of $H_{n+1}$ are $H_n \pm \alpha$ for some $\alpha \ge 0$. We show that these two values are sufficiently separated, specifically that $\alpha \ge \frac{3}{4} H_n(1-H_n)$. Thus, unless $H_n$ is very close to $0$ or $1$, the two new values have a sizeable difference. We use this to show that $\E[ \sqrt{H_n(1-H_n)} ]$ decreases by a constant factor in each step, which implies the desired exponential decay in $H_n(1-H_n)$ and therefore also $Z_n(1-Z_n)$. 
\footnote{We note that one can also prove directly
that $\E[\sqrt{Z_n (1-Z_n)}]$ decreases by a constant factor in each step and an earlier version of this paper (and independently \cite{hassani-thesis}) used this approach. The analysis presented here in terms of $H_n$ is cleaner and more intuitive in our opinion.}

The above bound is itself, however, not enough to prove
Theorem \ref{thm:main-intro}. What one needs is {\em fine
  polarization}, where the smallest $\approx I(W) N$ values
among $Z_n(i)$ all {\em add up} to a quantity that tends to $0$ for
large $N$ (in fact, this sum should be at most $2^{-N^{0.49}}$ if we
want the block error probability claimed in Theorem
\ref{thm:main-intro}).  To establish this, we use that in further steps, $Z_{n+1}$ reduces rapidly to $Z_n^2$ with probability $1/2$, together with Chernoff-bound
arguments (similar to \cite{arikan-telatar}) to bootstrap the rough
polarization of the first step to a fine polarization that suffices to bound the block error decodability.

Our analysis is elementary and self-contained, and does not use the
martingale convergence theorem.  The ingredients in our analysis were
all present explicitly or implicitly in various previous
works. However, it appears that their combination to imply a
polynomial convergence to capacity has not been observed
before, as evidenced by the explicit mention of this as an open
problem in the literature, eg. \cite[Section 6.6]{korada-thesis}, \cite[Section Ia]{KRU-spatial-ldpc}, \cite[Section 1.3]{shpilka12}, and \cite[Section I]{tal-vardy} (see the discussion following Corollary 2).

\subsection{Related work}
\label{sec:related-work}
The simplicity and elegance of the construction of polar codes, and
their wide applicability to a range of classic information theory
problems, have made them a popular choice in the recent
literature. Here we only briefly discuss aspects close to our focus on
the speed of polarization.

Starting with Ar\i kan's original paper, the ``rate of polarization''
has been studied in several works. However, this refers to something
different than our focus; this is why we deliberately use the term
``speed of polarization'' to refer to the question of how large $n$
needs to be before, say, $Z_n$ is in the range $(\eps,1-\eps)$ with
probability $\eps$. The rate of polarization refers to pinpointing a
function $\Upsilon$ with $\Upsilon(n) \to 0$ for large $n$ such that
$\lim_{n \to \infty} \mathrm{Pr}[ Z_n \le \Upsilon(n) ] = I(W)$. Ar\i
kan proved that one can take $\Upsilon(n) =
O(2^{-5n/4})$~\cite{arikan-polar}, and later Ar\i kan and Telatar
established that one can take $\Upsilon(n) = 2^{-2^{\beta n}}$ for any
$\beta < 1/2$~\cite{arikan-telatar}. Further they proved that for
$\gamma > 1/2$, $\lim_{n \to \infty} \mathrm{Pr}[ Z_n \le
  2^{-2^{\gamma n}} ] = 0$. This determined the rate at which the
Bhattacharyya parameters of the ``noiseless'' channels polarize to $0$
in the limit of larger $n$.  More fine grained bounds on this
asymptotic rate of polarization as a function of the code rate were obtained
in \cite{hassani-urbanke-i}.

For our purpose, to get a finite-length statement about the
performance of polar codes, we need to understand the speed at which
$\mathrm{Pr}[ Z_n \le \Upsilon(n)] $ approaches the limit $I(W)$ as
$n$ grows (any function $\Upsilon$ with $\Upsilon(n) = o(1/2^n)$ will
do, though we get the right $2^{-2^{0.49 n}}$ type decay).

Restated in our terminology, in \cite{GHU-isit12} the authors prove
the following ``negative result'' concerning gap to capacity: for
polar coding with successive cancellation (SC) decoding to have
vanishing decoding error probability at rates within $\eps$ of
capacity, the block length has to be {\em at least}
$(1/\eps)^{3.553}$. (A slight caveat is that this uses the sum
  of the error probabilities of the well-polarized channels as a
  proxy for the block error probability, whereas in fact this sum is
  only an upper bound on the decoding error probability of the SC
  decoder.)

Also related to the gap to capacity question is the work on ``scaling
laws,'' which is inspired by the behavior of systems undergoing a
phase transition in statistical physics. In coding theory, scaling
laws were suggested and studied in the context of iterative decoding
of LDPC codes in \cite{AMRU-scaling}. In that context, for a channel
with capacity $C$, the scaling law posits the existence of an exponent
$\mu$ such that the block error probability $P_e(N,R)$ as a function
of block length $N$ and rate $R$ tends in the limit of $N \to \infty$
while fixing $N^{1/\mu} (C-R) = x$, to $f(x)$ where $f$ is some
function that decreases smoothly from $1$ to $0$ as its argument
changes from $-\infty$ to $+\infty$. Coming back to polar codes, in
\cite{KMTU}, the authors make a {\em Scaling Assumption} that the
probability $Q_n(x)$ that $Z_n$ exceeds $x$ is such that $\lim_{n \to
  \infty} N^{1/\mu} Q_n(x)$ exists and equals a function $Q(x)$. Under
this assumption, they use simulations to numerically estimate $\mu
\approx 3.627$ for the BEC. Using the small $x$ asymptotics of $Q(x)$
suggested by the numerical data, they predict an $\approx
(1/\eps)^\mu$ upper bound on the block length as a function of the gap
$\eps$ to capacity for the BEC. For general channels, under the
heuristic assumption that the densities of log-likelihood ratios
behave like Gaussians, an exponent of $\mu \approx 4.001$ is suggested
for the Scaling Assumption. However, to the best of our knowledge, it
does not appear that one can get a rigorous upper bound on block
length $N$ as a function of the gap to capacity via these methods.

\section{Preliminaries}
We will work over a binary input alphabet $\B = \{0, 1\}$. Let $W : \B \to \Y$ be a binary-input memoryless symmetric channel with finite output alphabet $\Y$ and transition probabilities $\{W(y | x) : x \in \B, y \in \Y\}$. A binary-input channel is symmetric when the two rows of the transition probability matrix are permutations of each other; i.e., there exists a bijective mapping $\sigma : \Y \mapsto \Y$ where $\sigma = \sigma^{-1}$ and $W(y|0) = W(\sigma(y)|1)$ for all $y$. Both the binary erasure and binary symmetric channels are examples of symmetric channels.

Let $X$ represent a uniformly distributed binary random variable, and let $Y$ represent the result of sending $X$ through the channel $W$.

The entropy of the channel $W$, denote $H(W)$, is defined as the entropy of $X$, the input, given the output $Y$, i.e., $H(W) = H(X | Y)$. It represents how much uncertainty there is in the input of the channel given the output of the channel. The mutual information of $W$, sometimes known as the capacity, and denoted $I(W)$, is defined as the mutual information between $X$ and $Y$ when the input distribution $X$ is uniform:
\begin{dmath*}
  I(W)= I(X;Y)
  = 1  - H(X|Y) = 1 - H(W) \ . 
\end{dmath*}
We have $0 \le I(W) \le 1$, with a larger value meaning a less noisy channel.
As the mutual information expression is difficult to work with directly, we will often refer to the Bhattacharyya parameter of $W$ as a proxy for the quality of the channel:
\[
  Z(W) = \sum_{y \in \Y} \sqrt{W(y|0) W(y|1)} \ .
\]
This quantity is a natural one to capture the similarity between the channel outputs when the input is $0$ and $1$: $Z(W)$ is simply the dot product between the unit vectors obtained by taking the square root of the output distributions under input $0$ and $1$ (which is also called the Hellinger affinity between these distributions).

Intuitively, the Bhattacharyya parameter $Z(W)$ should be near $0$ when $H(W)$ is near $0$ (meaning that it is easy to ascertain the input of a channel given the output), and conversely, $Z(W)$ is near $1$ when $H(W)$ is near $1$. This intuition is quantified  by the following expression (where the upper bound is from \cite[Lemma 1.5]{korada-thesis} and the lower bound is from \cite{arikan-source}):
\begin{equation}
  \label{eq:ZversusH}
  Z(W)^2 \le H(W) \le Z(W) \ .
\end{equation}

Given a single output $y \in \Y$ from a channel $W$, we would like to be able to map it back to $X$, the input to the channel. The most obvious way to do this is by using the maximum-likelihood decoder:
\[
  \hat X = \argmax_{x \in \B} \Pr(x | y) = \argmax_{x \in \B} W(y|x)
\]
where a decoding error is declared if there is a tie. Thus, the probability of error for a uniform input bit under maximum likelihood decoding is
\begin{align*}
  P_{e}(W)& = \Pr(\hat X \ne X) \\
          &= \frac 12 \sum_{x \in \B} \sum_{y \in \Y} W(y|x) \1_{W(y|x) \le W(y|x\oplus 1)}
\end{align*}
where $\1_x$ denotes the indicator function of $x$. Directly from this expression, we can conclude
\begin{equation}
  P_e(W) \le Z(W)
  \label{eq:P_e<=Z}
\end{equation}
since $\1_{W(y|x) \le W(y|x \oplus 1)} \le \sqrt{W(y|x \oplus 1)}/\sqrt{W(y|x)}$, and the channel is symmetric (so the sum over $x \in \B$ and the $1/2$ cancel out). Thus, the Bhattacharyya parameter $Z(W)$ also bounds the error probability of maximum likelihood decoding based on a single use of the channel $W$. 

\section{Polar codes}
\label{sec:polar-prelims}
\subsection{Construction preliminaries}
This is a short primer on the motivations and techniques behind polar coding, following \cite{arikan-polar, sasoglu-book}. Consider a family of invertible linear transformations $G_n : \B^{2^n} \to \B^{2^n}$ defined recursively as follows: $G_0 = [ 1 ]$ and for a $2N$-bit vector $u = (u_0, u_1, \dots, u_{2N-1})$ with $N=2^n$, we define
\begin{dmath}
G_{n+1} u = \,  G_{n} (u_0 \oplus u_1, u_2 \oplus u_3, \dots, u_{2N-2} \oplus u_{2N - 1}) \circ
    G_n (u_1, u_3, u_5, \dots, u_{2N-1})
   \label{eq:Grecursion}
\end{dmath}
where $\circ$ is the vector concatenation operator. More explicitly, this construction can be shown to be equivalent to the explicit form $G_n =  K^{\otimes n} B_n$ (see \cite[Sec. VII]{arikan-polar}) where $B_{n}$ is the $2^n \times 2^n$ bit-reversal permutation matrix for $n$-bit strings, 
\(K = 
\left[
  \begin{matrix}
    1 & 1\\
    0 & 1
  \end{matrix}
\right]\) and $\otimes$ denotes the Kronecker product.

Suppose we use the matrix $G_n$ to encode a $N = 2^n$-size vector $U$, $X = G_n U$, and then transmit $X$ over a binary symmetric channel $W$. It can be shown with a Martingale Convergence Theorem-based proof \cite{arikan-polar} that for all $\eps > 0$,
\begin{equation}
  \lim_{N \to \infty} \Pr_i \left[ H(U_i | U_0^{i-1}, Y_0^{N-1})  < \epsilon \right] = I(W).
  \label{eq:asymptotics}
\end{equation}
where the notation $U_i^j$ denotes the subvector $(U_i, U_{i+1}, \dots, U_j)$.

In words, there exists a {\it good set} of indices $i$ so that for all elements in this set, given all of the outputs from the channel and (correct) decodings of all of the bits indexed less than $i$, the value of $U_i$ can be ascertained with low probability of error (as it is a low-entropy random variable). 

For every element that is outside of the good set, we do not have this guarantee; this suggests a encoding technique wherein we ``freeze'' all indices outside of this good set to a certain predefined value ($0$ will do). We call the indices that are not in the good set as the {\em frozen} set.
\subsection{Successive cancellation decoder}
The above distinction between good indices and frozen indices suggests a successive cancellation decoding technique where if the index is in the good set, we output the maximum-likelihood bit (which has low probability of being wrong due to the low entropy) or if the index is in the frozen set, we output the predetermined bit (which has zero probability of being incorrect). A sketch of such a successive cancellation decoder is presented in Algorithm \ref{alg:sc-decoder}.
\begin{definition}
  A polar code with frozen set $F \subset \{0,1,\dots,N-1\}$ is defined as
  \[ C_F = \{G_n u \mid u \in \{0,1\}^N, u_F = 0\} \ . \]
\end{definition}

  \label{sec:sc-decoder}
\begin{algorithm}[h]
  \caption{Successive cancellation decoder\label{alg:sc-decoder}}
  \SetKwInOut{Input}{input}
  \SetKwInOut{Output}{output}
  \LinesNotNumbered 
  \DontPrintSemicolon

  \Input{$y_0^{N-1}$, $F$, $W$}
  \Output{$u_0^{K-1}$}
  \nl $\hat u \gets$ zero vector of size $N$ \;
  \nl \For{$i \in 0..N-1$} {
    \nl \uIf{$i \in F$}{
      \nl  $\hat u_i \gets 0$
    }
    \nl \Else{
      \nl \uIf{$\frac{\Pr(U_i = 0 | U_0^{i-1} = \hat u_0^{i-1}, Y_0^{N-1} = y_0^{N-1})}{\Pr(U_i = 1 | U_0^{i-1} = \hat u_0^{i-1}, Y_0^{N-1} = y_0^{N-1})} > 1$} {  \label{line:prob-calc}
        \nl $\hat u_i \gets 0$
      }
      \nl \Else {
        \nl $\hat u_i \gets 1$
      }
    }
  }
  \nl \Return $\hat u_{\overline{F}}$

\medskip\noindent  {\bf Remark}. The probability ratio on line $\ref{line:prob-calc}$ can be computed with a na\"ive approach by recursively computing (where $n = \lg N$) $W_n^{(i)}(y_0^{N-1}, \hat u_0^{i-1}|x)$ for $x \in \B$ according to the recursive evolution equations \eqref{eq:w-},\eqref{eq:w+},\eqref{eq:evolution}. The result is true if the expression is larger for $x = 1$ than it is for $x = 0$, as by Bayes's theorem,
\begin{align*}
    \Pr(U_i = 0 | U_0^{i-1} = \hat u_0^{i-1}, Y_0^{N-1} = y_0^{N-1}) =\\
    \frac{W_n^{(i)}(y_0^{N-1}, \hat u_0^{i-1} | 0) \Pr(u_i = 0)}{\Pr(U_0^{i-1} = \hat u_0^{i-1}, Y_0^{N-1} = y_0^{N-1})},
  \end{align*}
    and the term in the denominator is present in both the $U_i = 0$ and $U_i = 1$ expression and therefore cancels in the division; the $\Pr(u_i = 0)$ term cancels as well for a uniform prior on $u_i$ (which is necessary to achieve capacity for the symmetric channel $W$).
  
    The runtime of the algorithm can be improved to $\O(N \log N)$ by computing the probabilities on line \ref{line:prob-calc} with a divide-and-conquer approach as in \cite{arikan-polar}. We note that this runtime bound assumes constant-time arithmetic; consideration of $n$-bit arithmetic relaxes this bound to $\O(N \polylog(N))$. For a treatment of more aggressive quantizations, see \cite[Chapter 6]{hassani-thesis}.
\end{algorithm}

By \eqref{eq:asymptotics}, if we take $F$ to be the positions with conditional entropy exceeding $\eps$, the rate of such a code would approach $I(W)$ in the limit $N \to \infty$.

To simplify the probability calculation (as seen on line $\ref{line:prob-calc}$ of Algorithm \ref{alg:sc-decoder} and explained further in the comments), it is useful to consider the induced channel seen by each bit, $\wni : \B \to \Y^{N} \times \B^{i}$, for $0 \le i \le 2^n-1$. Here, we are trying to ascertain the most probable value of the input bit $U_i$ by considering the output from all channels $Y_0^{N-1}$ and the (decoded) input from all channels before index $i$. Since the probability of decoding error at every step is bounded above by the corresponding Bhattacharyya parameter $Z$ by \eqref{eq:P_e<=Z}, we can examine $Z(\wni)$ as a proxy for $P_e(\wni)$.

It will be useful to redefine $\wni$ recursively both to bound the evolution of $Z(\wni)$ and to facilitate the computation. Consider the two transformations $^-$ and $^+$ defined as follows:
\begin{equation}
  W^-(y_1, y_2|x_1) = \sum_{x_2 \in \B} \frac12 W(y_1 | x_1 \oplus x_2) W(y_2|x_2)
  \label{eq:w-}
\end{equation}
and
\begin{equation}
  W^+(y_1, y_2, x_1 | x_2) = \frac12 W(y_1 | x_1 \oplus x_2) W(y_2 | x_2).
  \label{eq:w+}
\end{equation}
This process \eqref{eq:w-} and \eqref{eq:w+} preserves information in the sense that
\begin{equation}
  \label{eq:I-conserved}
  I(W^-) + I(W^+) = 2I(W),
\end{equation}
which follows by the chain rule of mutual information, as (suppose $X_1$ is the input seen at $W^-$ and $X_2$ is the input seen at $W^+$ and $Y_1, Y_2$ are the corresponding output variables)
\begin{align*}
  I(W^-) + I(W^+) &= I(X_1; Y_1, Y_2) + I(X_2; Y_1, Y_2 | X_1)\\
                  &= I(X_1, X_2; Y_1, Y_2) = 2I(W).
\end{align*}
We also associate $^-$ with a ``downgrading'' transformation and $^+$ with an ``upgrading'' transformation, as $I(W^-) \le I(W) \le I(W^+)$. 

Tying the operations $^-$ and $^+$ back to $Z(\wni)$, we notice that $W^- = W_1^{(0)}$ (the transformation $^-$ adds uniformly distributed noise from another input $x_2$, which is equivalent to the induced channel seen by the $0$th bit) and $W^+ = W_1^{(1)}$ (where here we clearly have the other input bit). More generally, by the recursive construction \eqref{eq:Grecursion}, one can conclude that the $\wni$ process can be redefined in a recursive manner as
\begin{equation}
  W_{n+1}^{(i)} = 
  \begin{cases}
    \left( W_n^{(\lfloor i/2\rfloor)} \right)^- & \mbox{ if $i$ is even} \\
    \left( W_n^{(\lfloor i/2 \rfloor)} \right)^+ & \mbox{ if $i$ is odd} 
  \end{cases}
  \label{eq:evolution}
\end{equation}
with the base channel $W_0^{(0)} = W$.

The evolution of $I(W^+)$ and $I(W^-)$ is difficult to analyze, but we will see in the next section that we can adequately bound $Z(W^+)$ and $Z(W^-)$ as a proxy. Such bounds are sufficient for analyzing our decoder, as we can bound the block error probability obtained by the successive cancellation decoder described in algorithm \ref{alg:sc-decoder} with bounds on the Bhattacharyya parameters of the subchannels. The probability of the $i$th (not frozen) bit being misdecoded by the algorithm, given the channel outputs and the input bits with index less than $i$, is bounded above by $Z(\wni)$ by equation \eqref{eq:P_e<=Z}. This observation, with the union bound, immediately gives the following lemma.

\begin{lemma}
  The block error probability of Algorithm \ref{alg:sc-decoder} on a polar code $C$ of length $n$ with frozen set $F$ is bounded above by the sum of the Bhattacharyya parameters $\sum_{i \in \overline{F}} Z(W_n^{(i)})$.
  \label{l:sc-prob}
\end{lemma}

\subsection{Bounds on $Z(W^-)$ and $Z(W^+)$}
A proof of these bounds can be found in \cite{arikan-polar,korada-thesis}, and the results are rederived in Appendix \ref{app:Z-evolve} for clarity and completeness. 
\begin{proposition}
$
  Z(W^+) = Z(W)^2
$
for all binary symmetric channels $W$.
\label{p:zplus}
\end{proposition}

\begin{proposition}
\(\displaystyle
  Z(W^-) \le 2Z(W) - Z(W)^2
\)
for all binary symmetric channels $W$, with equality if the channel $W$ is an erasure channel.
\label{p:zminus}
\end{proposition}
\section{Speed of polarization}
\label{s:speed}
Our first goal is to show that for some $m = \O(\log (1/\epsilon))$, we have that $\Pr_i[Z(W^{(i)}_m) \le 2^{-\O(m)}] \ge I(W) - \epsilon$ (the channel is ``roughly'' polarized). We will then use this rough polarization result to show that, for some $n = \O(\log(1/\epsilon))$, ``fine'' polarization occurs: $\Pr_i[Z(W^{(i)}_n) \le 2^{-2^{\beta n}}] \ge I(W) - \epsilon$. This approach is similar to the bootstrapping method used in \cite{arikan-telatar-arxiv}.
  
\subsection{Rough polarization}
We give a formal statement of rough polarization in the proposition below. A similar statement can be constructed for binary erasure channels (as opposed to general symmetric channels) with a much simpler proof; we include the statement and the simpler analysis in Appendix \ref{s:binary-erasure-rough}

\begin{proposition}
\label{p:bsc-rough}
There is an absolute constant $\Lambda < 1$ such that the following holds. For all $\rho \in (\Lambda, 1)$, there exists a constant $c_\rho$ such that for all binary-input symmetric channels $W$, all $\eps > 0$ and $m \ge b_\rho \log(1/\epsilon)$, there exists a {\em roughly polarized} set 
\begin{equation}
  \label{eq:W_r}
  \W_r \subset \W \triangleq \{\wmi : 0 \le i \le 2^{m}-1\}
\end{equation}such that for all $M \in \W_r$, $Z(M) \le 2 \rho^m$ and $\Pr_i(W^{(i)}_m \in \W_r) \ge I(W) - \epsilon.$
\end{proposition}

We first offer the following quantitative bound on the evolution of each step of the polarization process.
\begin{lemma}
  For all channels $W$, we have
    $H(W^+) \le H(W) - \alpha(W)$
  and
    $H(W^-) \ge H(W) + \alpha(W)$
  for $\alpha(W) = \theta H(W) (1-H(W))$, where $\theta$ is a constant greater than $3/4$.
  \label{l:h-a-bound}
  \label{l:thetabound}
\end{lemma}
\begin{proof}
  Let \[
    \hat \theta = \inf_{W} \frac{ H(W^-) - H(W)  }{H(W) (1 - H(W))},
  \]
  where the minimization is done over all binary-input symmetric channels $W$. Expanding the definition of the $^-$ transform, obtain
  \begin{equation}
    H(W^-) - H(W) = H(X_1 + X_2|Y_1, Y_2) - H(X_1|Y_1)
    \label{eq:expand-channels}
  \end{equation}
  where $X_1, X_2$ are uniformly distributed random bits, $Y_1$ and $Y_2$ in the first expression are distributed according to the transition probabilities of $W^-$ and $Y_1$ is distributed according to the transition probabilities of $W$.

  \cite[Lemma 2.2]{sasoglu-book} implies that if $(X_1, Y_1)$ and $(X_2, Y_2)$ are independent pairs of discrete random variables with $X_1, X_2 \in \B$ and $H(X_1|Y_1) = H(X_2|Y_2) = \alpha$, we have
  \[
    H(X_1 + X_2|Y_1, Y_2) - H(X_1 | Y_1) \ge \epsilon(\alpha),
  \]
  where $\epsilon(\alpha) = h(2h^{-1}(\alpha)(1-h^{-1}(\alpha))) - \alpha$ (here, $h$ is the binary entropy function and $h^{-1}$ is its inverse). Substituting the above in \eqref{eq:expand-channels}, we can write
  \begin{equation}
    H(W^-) - H(W) \ge h(2 h^{-1}(\alpha)(1-h^{-1}(\alpha))) - \alpha.
    \label{eq:hinv-sasoglu}
  \end{equation}

  We can therefore bound the desired expression by numerically minimizing the expression
  \begin{equation}
    \frac{h(2x(1-x)) - h(x)}{h(x)(1-h(x))}
    \label{eq:h-change-bound}
  \end{equation}
  over $x \in (0, 1/2)$ (the range of $h^{-1}$), which offers us $\hat \theta > .799$. We also derive an analytic bound on \eqref{eq:h-change-bound} in Appendix \ref{s:h-change-bound}.

  Since mutual information is conserved in our transformation (as stated in Equation \eqref{eq:I-conserved}), we can conclude the lemma, as any $\theta < \hat \theta$ suffices for the statement to be true.
\end{proof}
We define the symmetric entropy of a channel as
\[
  T(W) = H(W)(1-H(W)).
\]

To relate $T(\wni)$ back to $H(\wni)$, it is useful to define the sets (where $\rho \in (0,1)$): 
\begin{align*}
\label{eq:bad-good-sets}
 &A^g_\rho =  \left\{ i : H(\wni) \le \frac{1 -\sqrt{1 - 4 \rho^n}}{2} \right\},\\
 &A^b_\rho = \left\{ i : H(\wni) \ge \frac{1 + \sqrt{1 - 4 \rho^n}}{2} \right\} \ , ~~~\text{and}\\
 &A_\rho = A_\rho^g \cup A_\rho^b = \{ i : T(\wni) \le \rho^n \} \ .
\end{align*}
We associate $A_\rho^g$ with the ``good'' set (the set of $i$ such that the entropy, and therefore probability of misdecoding, is small) and $A_\rho^b$ with the ``bad'' set.
We record the following useful approximations, both of which follow from $\sqrt{1-4\rho^n} \ge 1-4\rho^n$.
\begin{fact}
\label{fact:wni-approx}
For $i \in A^g_\rho$, $H(\wni) \le 2 \rho^n$, and for $i \in A^b_\rho$, $H(\wni) \ge 1 - 2 \rho^n$.
\end{fact}
  We first state a bound on the evolution of $\sqrt{T(W_{n+1}^{(i)})}$.
  \begin{lemma}
    There exists a universal constant $\Lambda < 1$ such that
    \[
    \E_{i \bmod 2} \sqrt{T(W_{n+1}^{(i)})} \le \Lambda \sqrt{T(\wnip)} \ ,
    \]
    where the meaning of the expectation is that we fix $\lfloor i/2 \rfloor$ and allow $i \bmod 2$ to vary.
    \label{l:mbound}
  \end{lemma}

  \begin{proof}
    Defining $h = H(\wni)$, we have
    \begin{dmath}
      \E_{i \bmod 2} \sqrt{T(W_{n+1}^{(i)})} = \frac 12 \left( \sqrt{h(1-h) + (1-2h)\alpha - \alpha^2} + \sqrt{h(1-h) - (1-2h)\alpha - \alpha^2}\right)
      \label{eq:mbound-expr}
    \end{dmath}
    where $\alpha = H( (\wni)^-) - H(\wni) = H(\wni) - H( (\wni)^+)$. By symmetry, we can assume $h \le 1/2$ without loss of generality, and we also know that $\alpha \ge \theta h(1-h)$ from Lemma \ref{l:thetabound}. We can write
    \begin{align*}
      2\left(\E_{i \bmod 2} \sqrt{T(W_{n+1}^{(i)})}\right) & \le \sqrt{h(1-h) + (1-2h)\alpha} + \sqrt{h(1-h) - (1-2h)\alpha}\\
      &\le \sqrt{h(1-h)} - \frac{((1-2h)\alpha)^2}{4\left( (h(1-h) \right)^{3/2}} - O((1-2h)\alpha)^4 \\
      &\le \sqrt{h(1-h)} - \frac{( (1-2h) \theta h(1-h))^2}{4(h(1-h)^{3/2}}\\
      &= \sqrt{h(1-h)} - \frac{\theta^2}{4} (1-2h) \sqrt{h(1-h)}
    \end{align*}
  where the second line is a Taylor expansion around $h(1-h)$. This analysis gives the desired result for whenever $1 - 2h$ is greater than an absolute constant. For clarity of analysis, let us fix a concrete constant $1 - 2h \ge 1/100$.

  We can therefore focus on the case where $1 - 2 h < 1/100$, which implies $h \in [99/200, 1/2]$. Continuing, we have
  \begin{dmath}
    \alpha \ge \theta h(1-h) \ge 99\theta /400,
    \label{eq:alpha-away-from-zero}
  \end{dmath}an absolute constant bounded away from zero. We can also write $\alpha \ge 2(1-2h)$ as $2(1-2h) \le 1/50$, which is less than $99\theta/400$, since $\theta > 3/4$ from Lemma \ref{l:thetabound}. This expression implies
  \[
    (1-2h)\alpha - \alpha^2 \le  - \alpha^2/2
  \]
  which, when inserted into \eqref{eq:mbound-expr}, offers
  \[
    \E_{i \bmod 2} \sqrt{T(W_{n+1}^{(i)})} \le \sqrt{h(1-h)-\frac{\alpha^2}{2}}.
  \]
  This implies the existence of a $\Lambda < 1$, since $\alpha$ is bounded away from zero in \eqref{eq:alpha-away-from-zero}, $h(1-h) \le \frac{1}{4}$, and the function $\frac{\sqrt{x - c}}{\sqrt{x}}$ is increasing for positive $c$ and $x > c$.
  \end{proof}

     \begin{corollary}
     \label{cor:markov-general}
       Taking $\Lambda$ as defined in Lemma \ref{l:mbound}, $\Pr_i [T(W_n^{(i)}) \ge \alpha^n] \le \frac12 \left( \frac{\Lambda^2}{\alpha} \right)^{n/2}$
     \end{corollary}
      \begin{proof}
Clearly we have
    \[
      \E_i \sqrt{T(W_{n+1}^{(i)})} \le \Lambda^n \sqrt{T(W)} \le \Lambda^n \cdot \frac12
    \]
    and we can therefore use Markov's inequality to obtain the desired consequence.
      \end{proof}

      We are now in a position where we can conclude Proposition \ref{p:bsc-rough}.

\begin{proof}[Proof of Proposition \ref{p:bsc-rough}]
We have 
  \begin{align}
    \nonumber &\Pr(\overline{A_\rho}) \max_{i \in \overline{A_\rho}} (I(\wni)) + \Pr(A_\rho^b) \max_{i \in A_\rho^b} I(\wni) + \\
    &\Pr(A_\rho^g)\max_{i \in A_\rho^g} I(\wni) \ge \E_i(I(\wni)) = I(W)
    \label{eq:bsc-averaging}
  \end{align}
  where the last equality follows by the conservation of mutual information in our transformation as stated in equation \eqref{eq:I-conserved}.

  By definition, $I(W) = 1 - H(W)$. As $\min_{i \in A^b_\rho} H(\wni) \ge 1-2\rho^n$ by Fact \ref{fact:wni-approx}, we have $\max_{i \in A^b_\rho} I(\wni) \le 2 \rho^{n}$.
Using this together with equation \eqref{eq:bsc-averaging}, obtain
  \begin{equation*}
    \Pr(\overline{A_\rho}) + \Pr(A_\rho^b)\cdot  2 \rho^{n} + \Pr(A_\rho^g) \ge I(W)
  \end{equation*}
  where we used the trivial inequality (for binary-input channels) $I(\wni) \le 1$ for every $i$. Rearranging terms, using the bounds $\Pr(\overline{A_\rho}) \le \frac12 (\Lambda^2 / \rho)^{n/2}$ from Corollary \ref{cor:markov-general} and $H(\wni) \le 2\rho^n$ for $i \in A_\rho^g$ from Fact \ref{fact:wni-approx}, we get
  \begin{align}
    \nonumber \Pr_i [ H(\wmi) \le 2 \rho^m ] &\ge  \Pr(A_\rho^g) \\
                                   &\ge I(W) - \frac12 (\Lambda^2 / \rho)^{m/2} - 2 \rho^{m} \ .
    \label{eq:bsc-averaging2}
  \end{align}
  Clearly, if $\rho > \Lambda^2$, there is a constant $b_\rho$ such that $m \ge b_\rho \log (1/\epsilon)$ implies that the above lower bound is at least $I(W)-\eps$. We conclude our analysis by noting that $Z(W) \le \sqrt{H(W)}$, as observed in \eqref{eq:ZversusH}, so that $\sqrt{\rho} > \Lambda$ can play the role of $\rho$ for a lower bound similar to \eqref{eq:bsc-averaging2} on $\Pr_i [ Z(\wmi) \le 2 \kappa^m ]$ for $\kappa \in (\Lambda,1)$. 
\end{proof}
\subsection{Fine polarization}
The following proposition formalizes what we mean by ``fine polarization.''
\begin{proposition}
  \label{p:fine-polarization}
  Given $\delta \in (0, 1/2)$, there exists a constant $c_\delta$ for all binary input memoryless channels $W$ and $\epsilon > 0$ such that if $n_0 > c_\delta \log (1/\epsilon)$ then
  \[
    \Pr_{i} \left[Z( W_{n_0}^{(i)}) \le 2^{-2^{\delta n_0}}\right] \ge I(W) - \epsilon.
  \]
\end{proposition}
We will first need the following lemma to specify one of the constants.
\begin{lemma}
  \label{l:c_chi}
For all $\gamma > 0$, $\beta \in (0,1/2)$ and $\rho\in (0,1)$, there exists a constant $\theta(\beta, \gamma, \rho)$ such that for all $\eps \in (0,1)$, if $m > \theta(\beta, \gamma, \rho) \cdot \log(2/\epsilon)$, then
  \[
    \left( \frac{\lg(2/\rho)\gamma}{2} + 1 \right) \exp\left( -
      \frac{(1-2\beta)^2 \lg(2/\rho) m}{2} 
    \right) < \epsilon/2 \ .
  \]
\end{lemma}
\begin{proof}
  We can rewrite this expression as $c_1\exp(-c_2 m) < \epsilon$ for constants  $c_1, c_2$ that are independent of $\epsilon$ and the result is clear.
\end{proof}
\begin{proof}[Proof of Proposition \ref{p:fine-polarization}]
    Fix a $\beta \in (\delta, 1/2)$, and let $\gamma = \frac{\delta}{\beta - \delta}$. Let $\rho$ be an appropriate constant for Proposition \ref{p:bsc-rough}, and $b_\rho$ be the associated constant. We will define 
    \[ c_\delta = (1+\gamma) \max \{ 2 b_\rho, \theta(\beta, \gamma, \rho), 2/\rho, c_\beta \}, \]
      where $c_\beta$ is defined as a bound on $m$ such that if $m \ge c_\beta$, then 
      \[
        \frac{1-\beta}{1 - 2^{-\frac{n_0 - m}{c_\rho} \beta}} \le 1;
      \]
      $c_\beta = \frac{-\lg(\beta)}{2 \lg (2/\rho) \beta}$ suffices.
    
    Fix an $n_0 > c_\delta \log(1/\epsilon)$, $m = \frac{1}{1+\gamma} n_0$ and $n = n_0 - m = \gamma m$. We first start with a set of roughly polarized channels; by our choice of $c_\delta$, $m > b_\rho \log (2/\eps)$ and we can apply Proposition \ref{p:bsc-rough} and obtain a set $\W_r$ where
    \begin{equation}
      \label{eq:P(rough)}
      \Pr_i [Z_m^{(i)} \in \W_r] \ge I(W) - \epsilon/2   
    \end{equation}
    and $Z(M) \le 2 \rho^m$ for all $M \in \W_r$. Let $R(m)$ be the set of all associated indices $i$ in $\W_r$.
    
    Fix a $M \in \W_r$ and define a sequence $\{\tilde Z_n^{(i)}\}$ where
    \begin{equation}
    \tilde Z_{n+1}^{(i)} =
    \begin{cases}
      (\tilde Z_n^{(\lfloor i/2 \rfloor)})^2 & i \bmod 2 \equiv 1\\
      2\tilde Z_n^{(\lfloor i/2 \rfloor)} & i \bmod 2 \equiv 0\\
    \end{cases},
    \label{eq:fine-evolution}
  \end{equation}
  with the base case $Z_0^{(0)} = Z(M)$. Clearly $Z(M_n^{(i)}) \le \tilde Z_{n}^{(i)}$ by Proposition \ref{p:zminus}.  (Recall that $X_n^{(i)}$ is the polarization process done for $n$ steps with $i$ determining which branch to take for arbitrary binary input channel $X$.)

  Let $c_\rho = \lceil \frac{n}{2m \lg (2/\rho)} \rceil$. 
  Fix a $\beta \in (0, 1/2)$. 
  Define a collection of events $\{G_j(n) : 1 \le j \le c_\rho\}$:
  \begin{equation}
    \label{eq:defG}
    G_j(n) = \biggl\{ i : \sum_{\substack{k \in [jn/c_\rho,\\ (j+1)n/c_\rho)}} i_k \ge \beta n / c_\rho \biggr\};
  \end{equation}
  here, $i_k$ indicates the $k$'th least significant bit in the binary representation of $i$. Qualitatively speaking, $G_j$ occurs when the number of $1$'s in the $j$'th block is not too small. 
  
Since each bit of $i$ is independently distributed, we can apply the Chernoff-Hoeffding bound \cite{hoeffding} (with $p = 1/2$ and $\epsilon = 1/2 - \beta$) to conclude
    \begin{align}
      \nonumber \Pr_i(i \in G_j(n)) &\le 1 - \exp(-2 (1/2 - \beta)^2 n /(c_\rho))\\
      &= 1 -\exp((1-2\beta)^2 n / (2 c_\rho))
      \label{eq:chernoff}
    \end{align}
    for all $j \in [c_\rho]$. 

    Define 
    \begin{equation}
      \label{eq:G}
      G(n) = \bigcap_j G_j(n).
    \end{equation}
    Applying the union bound to $G(n)$ with \eqref{eq:chernoff}, obtain
    \begin{equation}
      \label{eq:unionbound}
      \Pr_i(i \in G(n)) \ge 1 - c_\rho \exp(-(1-2\beta)^2 n / (2 c_\rho)).
    \end{equation}

    Now we develop an upper bound on the evolution of $\tilde Z$ for each interval of $n/c_\rho$ squaring/doubling operations, conditioned on $i$ belonging to the 
    high probability set $G(n)$.
    
    Fix an interval $j \in \{0, 1, \dots, c_\rho\}$. By the evolution equations \eqref{eq:fine-evolution} and the bound provided by \eqref{eq:defG}, it is easy to see that the greatest possible value for $\tilde Z_{(j+1)n/c_\rho}$ is attained by $(1-\beta)n/c_\rho$ doublings followed by $\beta n/c_\rho$ squarings. Therefore,
    \begin{dmath*}
      \lg \tilde Z_{(j+1)n/c_\rho}^{\lfloor i / 2^{jn/c_\rho} \rfloor} \le 2^{\beta n/c_\rho} \left( (1 - \beta) n / c_\rho + \lg \tilde Z_{jn/c_\rho}^{\lfloor i / 2^{(j-1)n/c_\rho} \rfloor} \right).
    \end{dmath*}

    Cascading this argument over all intervals $j$, obtain
    \begin{align}
      \nonumber &\lg Z(M_n^{(i)}) \\
      \nonumber &\le \lg \tilde Z_n^{(i)}\\
                &\le 2^{n\beta} \lg Z(M) + \frac{n}{c_\rho} (1-\beta) (2^{\beta n/c_\rho} + 2^{2 \beta n/c_\rho} + \cdots + 2^{n \beta})\\
      \nonumber &\le 2^{n\beta} \lg Z(M) + \frac{n}{c_\rho} (1-\beta) \frac{2^{n \beta}}{1 - 2^{-\frac{n}{c_\rho}\beta}}\\
      \nonumber &= 2^{n \beta} \left( \lg Z(M) + \frac{n}{c_\rho} \frac{1 - \beta}{1 - 2^{-\frac{n}{c_\rho} \beta}} \right) \\
      \nonumber    &\le 2^{n \beta} (\lg Z(M) + n/c_\rho) \quad \text{as $m \ge c_\beta$}
                \intertext{As $M \in \W_r$, $Z(M) \le 2\rho^m$, and $n/c_\rho \le 2 m \lg(2/\rho)$, we can bound the above with}
      \label{eq:lg(Z(M))} &\le -  2^{n \beta} \lg(2/\rho) m \\
      \nonumber & \le - 2^{n \beta} \quad \text{as $m \ge 2/\rho$}
    \end{align}
 This shows that
    \[
      Z(W_{n_0}^{(i)}) \le 2^{-2^{\beta n}} = 2^{-2^{\delta n_0}} ,
    \]
    where the equality is due to the definition of $n$ and $m$, as long as the first $m$ bits of $i$ are in $R(m)$ and the last $n$ bits of $i$ are in $G(n)$. The former has probability at least $I(W)-\epsilon/2$ by \eqref{eq:P(rough)} and the latter has probability at least
    \begin{align*}
      &1 - c_\rho \exp(-(1-2\beta)^2 n / (2c_\rho))\\
      &\ge 1 - 
    \left( \frac{\lg(2/\rho)\gamma}{2} + 1 \right) \exp\left( -
      \frac{(1-2\beta)^2 \lg(2/\rho) m}{2} 
    \right) \\
    &\ge 1 - \epsilon/2
    \end{align*}
    by our choice of $c_\delta$ and Lemma \ref{l:c_chi}. 

    Putting the two together with the union bound, obtain
    \[
 \Pr_i \left[ Z(W_{n_0}^{(i)}) \le 2^{-2^{\delta n_0}} \right] \ge I(W) - \epsilon. \qedhere
    \]
  \end{proof}
The following corollary will be useful in the next section, where we will deal with an approximation to the Bhattacharyya parameter. It relaxes the conditions on the polarized set from Proposition \ref{p:bsc-rough}.
  \begin{corollary}
    \label{c:looser}
    Proposition \ref{p:fine-polarization} still holds with a modified roughly polarized set (recall the definition of the roughly polarized set $\W_r$ from equation \eqref{eq:W_r}) $\widetilde{\W_r}$ where
    $\widetilde{\W_r} \supset \W_r$ 
    and 
    $Z\left(\widetilde{\W_r}\right) \le \sqrt{3 \rho^m}$ (instead of $2 \rho^m$) with a modified constant $\widetilde{c_\delta}$.
  \end{corollary}

  \begin{proof}
    The changes that need to be made follow from Equation \eqref{eq:lg(Z(M))}, where $\lg Z(M)$ is used. With the extra square root, an extra factor of $1/2$ appears outside of the $\lg$, which means $c_\rho$ needs to be adjusted by a constant factor. In addition, $\lg(2/\rho)$ needs to be adjusted to $\lg(3/\rho)$, but this is also just a constant change.
  \end{proof}

\section{Efficient construction of polar codes}
\label{s:construction}
The construction of a polar code reduces to determining the frozen set of indices (the generator matrix then consists of columns of $G_n = K^{\otimes n} B_n$ indexed by the non-frozen positions). The core component of the efficient construction of a frozen set is estimating the Bhattacharyya parameters of the subchannels $W_n^{(i)}$. In the erasure case, this is simple because the evolution equation offered by Proposition \ref{p:zminus} is exact. In the general case, the na\"ive calculation takes too much time: $W_n^{(i)}$ has an exponentially large output alphabet size in terms of $N = 2^n$. 

Our goal, therefore, is to limit the alphabet size of $W_n^{(i)}$ while roughly maintaining the same Bhattacharyya parameter. With this sort of approach, we can select channels with relatively good Bhattacharyya parameters. The idea of approximating the channel behavior by degrading it via output symbol merging is due to \cite{tal-vardy} and variants of it were analyzed in \cite{PHTT}. The approach is also discussed in the survey \cite[Section 3.3]{sasoglu-book}. Since we can only achieve an inverse polynomial error in estimating the Bhattacharyya parameters with a polynomial alphabet, we use the estimation only up to the 
rough polarization step, and then use the explicit description of the subsequent good channels that is implicit in the proof of Proposition \ref{p:fine-polarization}. 

We note that revised versions of the Tal-Vardy work \cite{tal-vardy} also include a polynomial time algorithm for code construction by combining their methods with the analysis of \cite{PHTT}. However, as finite-length bounds on the speed of polarization were not available to them, they could not claim $\mathrm{poly}(N/\eps)$ construction time, but only $c_\eps N$ time for some unspecified $c_\eps$.

We will first state our binning algorithm, along with its properties, and then conclude the main theorem.

\subsection{Binning Algorithm}
For our binning, we deal with the marginal distributions of the input bit given an output symbol. A binary-input symmetric channel $W$ defines a marginal probability distribution $W(y|x)$. We invert this conditioning to form the expression
\[
  p(0|y) = \Pr_x(x = 0 | W(x) = y) = \frac 12 \frac{W(y|0)}{\Pr_x(W(x)=y)}
\]
for a uniformly distributed input bit $x$. In addition, we introduce the one-argument form
\[
  p(y) = \Pr_x(W(x) = y)
\]
for the simple probability that the output is $y$ given an uniformly distributed input bit $x$.

\begin{algorithm}[h]
  \caption{Binning algorithm\label{alg:binning}}
  \SetKwInOut{Input}{input}
  \SetKwInOut{Output}{output}
  \LinesNumbered
  \DontPrintSemicolon

  \Input{$W : \B \to \Y$, $k > 0$}
  \Output{$\widetilde W : \B \to \widetilde \Y$}
  Initialize new channel $\widetilde W$ with symbols $\tilde y_0, \tilde y_1 \dots \tilde y_{k}$ with $\widetilde W(\tilde y | x) = 0$ for all $\tilde y$ and $x \in \B$ \;
  \For{$y \in \Y$}{
    $p(0 | y) \gets \frac 12\frac{W(y|0)}{\Pr_x(W(x) = y)}$\;
    $\widetilde W(\tilde y_{\lfloor k p(0|y) \rfloor} | 0) \gets$
    $\widetilde W(\tilde y_{\lfloor k p(0|y) \rfloor} | 0) + W(y | 0)$\;
    $\widetilde W(\tilde y_{ \lfloor k p(0|y) \rfloor} | 1) \gets$
    $\widetilde W(\tilde y_{ \lfloor k p(0|y) \rfloor} | 1) + W(y | 1)$\;
  }
  \Return $\widetilde W$
\end{algorithm}
\begin{proposition}
  \label{p:approx-z}
  For a binary-input symmetric channel $W : \B \to \Y$ and all $k > 0$, there exists a channel $\widetilde W: \B \to \widetilde \Y$ such that 
  \[ H(W) \le H(\widetilde W) \le H(W) + 2\lg(k)/k, \quad |\widetilde \Y| \le k+1 \ , \]
and the channel transition probabilities, $\widetilde W(y|x)$, are computable, by Algorithm \ref{alg:binning}, in time polynomial in $|\Y|$ and $k$.
\end{proposition}

We will delay the proof of Proposition \ref{p:approx-z} to Appendix \ref{s:output-symbol-binning}, as the details are somewhat mechanical. We note that a slightly different binning strategy~\cite{tal-vardy} can achieve an approximation error of $O(1/k)$. We chose to employ a simple variant that still works for our purposes.

We will iteratively use the binning algorithm underlying Proposition \ref{p:approx-z} to select the best channels. The following corollary formalizes this.
\begin{corollary}
  \label{c:iterated-binning}
  Let $\widehat{\wni}$ indicate the result of using Algorithm \ref{alg:binning} after every application of the evolution Equations \eqref{eq:evolution}; that is, \[
    \widehat{\wni} = 
    \widetilde{
      \smash{
        \widetilde{
          \widetilde{W^+}^-
        }^{\iddots}
      }
      \rule{0pt}{1.8em}
    }
    \]
    where the $+$ or $-$ is chosen depending on the corresponding bit, starting from the least significant one, of the binary representation of $i \in \{0,1,\dots,2^n-1\}$. Then
    \[
      H(\wni) \le H\left(\widehat{\wni}\right) \le H(\wni) + \frac{2^{n+2}\lg(k)}{k} \ .
    \]
\end{corollary}
\begin{proof}
  The lower bound is obvious as the operation $\widetilde{\cdot}$ never decreases the entropy of the channel, as mentioned in the proof of Proposition \ref{p:approx-z}.

  For the upper bound, we'd like to consider the error expression summed over all $\wni$:
  \begin{dmath}
    \label{eq:error}
    \sum_{i=0}^{2^n-1} H\left( \widehat{\wni} \right) - \sum_{i=0}^{2^n-1} H(\wni) = \sum_{i=0}^{2^n-1} H\left( \widehat{\wni} \right) - 2^n H(W)
  \end{dmath}
  as $\E_{b \in \{+, -\}} H(W^b) = H(W)$ by \eqref{eq:I-conserved}. At every approximation stage, we have, from Proposition \ref{p:approx-z},
  \begin{dmath*}
    H\left(\widetilde{\widehat{W_m^{(\lfloor i/2\rfloor)}}^+}\right)
    +
    H\left(\widetilde{\widehat{W_m^{(\lfloor i/2\rfloor)}}^-}\right)
    \le 
    2\left( 
    H\left( \widehat{W_m^{(\lfloor i/2\rfloor)}}\right) + \frac{2 \lg k}{k}
    \right) 
    \ .
  \end{dmath*}
  Applying this to every level of the expression \eqref{eq:error} (colloquially speaking, we strip off the $\widetilde{\hspace{1em}}$s $n$ times), obtain
  \begin{align*}
    \sum_{i=0}^{2^n - 1} H\left( \widehat \wni \right) - 2^n H(W) &\le \frac{2 \lg k}{k} (2 + 2^2 + \cdots + 2^n) \\
    &\le \frac{2^{n+2} \lg k}{k} \ .
  \end{align*}
   Since the sum of all of the errors $H\left( \widehat{\wni} \right) - H(\wni)$ is upper bounded by $\frac{2^{n+2} \lg k}{k}$, each error is also upper bounded by $\frac{2^{n+2} \lg k}{k}$ (since no error is negative due to the lower bound).
\end{proof}

We are now in a position to restate and prove our main theorem (Theorem \ref{thm:main-intro}).

\begin{theorem*}
There is an absolute constant $\mu < \infty$ such that the following
holds.  Let $W$ be a binary-input output-symmetric memoryless channel with
capacity $I(W)$. Then there exists $a_W < \infty$ such that for all
$\eps > 0$ and all powers of two $N \ge a_W (1/\eps)^\mu$, there is a
deterministic $\mathrm{poly}(N)$ time construction of a binary linear
code of block length $N$ and rate at least $I(W)-\eps$ and a
deterministic $\O(N \log N)$ time decoding algorithm with block
error probability at most $2^{-N^{0.49}}$.
\end{theorem*}

\begin{proof}
  Fix an $N$ that is a power of $2$, and let $n_0 = \lg(N)$. Define $m, n, \rho$ as they are defined in the proof of Proposition \ref{p:fine-polarization}. Utilizing the definition of $\ \widehat{\cdot}\ $ from Corollary \ref{c:iterated-binning} with $k = \left( \frac {2} {\rho} \right)^{2m}$, let $\widehat{\W_r}$ be the set of all channels $W_m^{(i)}$ such that $H\left(\widehat{W_m^{(i)}}\right) \le 3 \rho^m$, and let $\widehat{R}(m)$ be the set of corresponding indices $i$. Define the complement of the frozen set
  \[
    \overline{\hat F_{n_0}} =
     \left\{ 
       i ~ \left| ~ 
         \begin{aligned}
                  &0 \le i \le 2^{n_0} - 1, \\
                  &i_{0}^{m-1} \in \hat R(m),\\
                  &i_{m}^{n_0 - 1} \in G(n_0 - m) 
         \end{aligned}
         \right.
     \right\}
  \]
  where $G(n)$ is defined in Equation \ref{eq:G} and the notation $i_j^k = i / 2^j \bmod 2^{k-j+1}$ means the integer with the binary representation of the $j$th through $k$th bits of $i$, inclusive. We note that this set $\overline{\hat F_{n_0}}$ is computable in $\poly(1/\eps, N)$ time: $\widehat R(m)$ is computable in $\poly(1/\eps)$ time because $k \le \poly(1/\eps)$ and $G(n_0 - m)$ is computable in $\O(N)$ time as it is just counting the number of $1$ bits in various intervals.

  By Corollary \ref{c:iterated-binning} we can conclude that $i \in R(m)$ implies $i \in \hat R(m)$ because $Z(\wmi) \le 2\rho^m$ implies $H(\wmi) \le Z(\wni) \le 2\rho^m$. This in turn implies $H(\widehat{\wmi}) \le 3 \rho^m$ by our choice of $k$ and the approximation error guaranteed by Corollary \ref{c:iterated-binning}. Therefore, we have
  \[
    \Pr_{i < 2^m} (i \in \hat R(m)) \ge \Pr_{i < 2^m} (i \in R(m))
  \]
  and also that all $M \in \widehat{\W_r}$ satisfy
  $Z(M) \le \sqrt{H(M)}  \le \sqrt{3 \rho^m}$, where the former inequality is from \eqref{eq:ZversusH}.
 
  Applying Corollary \ref{c:looser} with our modified set $\widehat{\W_r}$, we can now conclude 
  $\Pr(i \in \overline{\hat F_{n_0}}) \ge I(W) - \epsilon$ 
  and
  $Z(W_{n_0}^{(i)}) \le 2^{-2^{\delta n_0}}$ 
  for all $i$ in $\overline{\hat F_{n_0}}$. This implies that 
  \[
    \sum_{i \in \overline{\hat F_{n_0}}} Z(\wni) \le N2^{-N^\delta} \ .
  \]
  Taking $\delta = .499$ and $\mu = \widetilde{c_\delta}$, we can conclude the existence of 
  an $a_W$ such that for $N \ge a_W (1/\epsilon)^\mu$, 
  \[
    \sum_{i \in \overline{\hat F_{n_0}}} Z(\wni) \le 2^{-N^{.49}}, 
    \]
    as such $\mu$ satisfies the conditions of Corollary \ref{c:looser}. The proof is now complete since by Lemma \ref{l:sc-prob}, the block error probability of polar codes with a frozen set $F$ under successive cancellation decoding is bounded by the sum of the Bhattacharyya parameters of the channels not in $F$.
\end{proof}

\section{Future work}
The explicit value of $\mu$ found in Theorem \ref{thm:main-intro} is a large constant and far from the empirically suggested bound of approximately $4$. Tighter versions of this analysis should be able to minimize the difference between the upper bound suggested by Theorem \ref{thm:main-intro} and the available lower bounds.

We hope to extend these results shortly to channels with non-binary input alphabets, utilizing a decomposition of channels to prime input alphabet sizes \cite{polar-prime-alphabets}. Another direction is to study the effect of recursively using larger $\ell \times \ell$ kernels instead of the $2 \times 2$ matrix $K =\left( \begin{smallmatrix} 1 & 1 \\ 0 & 1 \end{smallmatrix} \right)$. Of course in the limit of $\ell \to \infty$, by appealing to the behavior of random linear codes we will achieve $\mu \approx 2$, but the decoding complexity will grow as $2^\ell$. The trade-off between $\mu$ and $\ell$ for fixed $\ell > 2$ might be interesting to study.

\section*{Acknowledgments}
We thank Seyed Hamed Hassani, Eren \c{S}a\c{s}o\u{g}lu, Madhu Sudan, Ido Tal, R\"{u}diger Urbanke, and Alexander Vardy for useful discussions and comments about the write-up.

\appendix
\section{Proofs of $Z$-parameter evolution equations}
\label{app:Z-evolve}
The $Z$-parameter evolution equations are a special case of the lemmas in \cite{korada-thesis}, specifically in the appendices to Chapters 2 and 3, and the proof techniques used here are based on the proofs of those lemmas.

\begin{proof}[Proof of Proposition \ref{p:zplus}]
  This can be done directly by definition. Let $\Y$ be the output alphabet of $W_n$. Then
\begin{align*}
  Z(W_n^+) &\triangleq \sum_{y \in \B \times \Y^2} \sqrt{W_n^+ (y | 0) W_n^+(y | 1)}\\
  &= \frac 12 \sum_{x \in \B, y_1, y_2 \in \Y} \sqrt{ W_n (y_1 | x \oplus 0) W_n(y_2 |0) W_n (y_1|x \oplus 1) W_n(y_2 | 1)}\\
  &= \frac 12 \sum_{x \in \B, y_1 \in \Y} \sqrt{W_n(y_1 | x) W_n(y_1 | x \oplus 1)} \sum_{y_2 \in \Y} \sqrt{W_n(y_2 | 0) W_n(y_2 | 1)}\\
  &= \sum_{y_1 \in \Y} \sqrt{W_n (y_1 | 0) W_n(y_1 | 1)} \sum_{y_2 \in \Y} \sqrt{W_n (y_2 | 0) W_n(y_2 | 1)}\\
  &\triangleq Z(W_n)^2
\end{align*}
where the first step is the expansion of the definition of $W_n^+$ and the rest is arithmetic. 
\end{proof}

\begin{proof}[Proof of Proposition \ref{p:zminus}]
We first show $Z(W_n^-) \le 2Z(W_n) - Z(W_n)^2$. Again, let $\Y$ be the output alphabet of $W_n$. Then we have
\begin{align}
  Z(W_n^-) &\triangleq \sum_{y_1, y_2 \in \Y} \sqrt{W_n^-(y_1, y_2 | 0) W_n^-(y_1, y_2 |1)}\\
  \nonumber &= \frac12 \sum_{y_1, y_2 \in \Y} \sqrt{
    \sum_{x_1 \in \B} W_n(y_1 | x_1) W_n(y_2 | x_1)
    \sum_{x_2 \in \B} W_n(y_1 | 1 \oplus x_2) W_n(y_2 | x_2)
  }\\
  \nonumber &= \frac12 \sum_{y_1, y_2 \in \Y} 
  \sqrt{(W_n(y_1|0)W_n(y_1|1))}
  \sqrt{(W_n(y_2|0)W_n(y_2|1))}\\
  &\hspace{3em}
  \sqrt{
    \frac{W_n(y_1|0)}{W_n(y_1|1)} + 
    \frac{W_n(y_2|0)}{W_n(y_2|1)} +
    \frac{W_n(y_1|1)}{W_n(y_2|0)} +
    \frac{W_n(y_2|1)}{W_n(y_2|0)}
  }
  \intertext{and we note that we can define a probability mass function $p(y) = \frac{\sqrt{W_n(y|0)W_n(y|1)}}{Z(W_n)}$ over $\Y$, so we write}
  \nonumber &= \frac{Z(W_n)^2}{2} \sum_{y_1, y_2 \in \Y} p(y_1) p(y_2) 
  \sqrt{
    \frac{W_n(y_1|0)^2 + W_n(y_1|1)^2}{W_n(y_1|0)W_n(y_1|1)} + 
    \frac{W_n(y_2|0)^2 + W_n(y_2|1)^2}{W_n(y_2|0)W_n(y_2|1)}
  }
  \intertext{and introducing $f(y) = \sqrt{W_n(y|0)/W_n(y|1)} + \sqrt{W_n(y|1)/W_n(y|0)}$, we can write}
  \label{eq:E(f)} &= \frac{Z(W_n)^2}{2} \E_{y_1, y_2 \sim p(y)} \sqrt{f(y_1)^2 + f(y_2)^2 - 4}\\
  \nonumber &\le \frac{Z(W_n)^2}{2} \left( \E(f(y_1)) + \E(f(y_2)) - 2 \right) \quad 
  \text{using $\sqrt{a+b-c} \le \sqrt{a} + \sqrt{b} - \sqrt{c}$ when $a,b \ge c$} \\
  \intertext{and since $\E_{y_1 \sim p(y)} [ f(y_1)] = 2/Z(W_n)$,}
  \nonumber &= 2Z(W_n) - Z(W_n)^2
 \end{align}
 We note that $p(y) = 0$ for all $y$ where either $W_n(y|0)$ or $W_n(y|1)$ is zero, so the expressions involving $f(y)$ are well-defined even if $f(y)$ is not defined for all $y$. 
In the case that $W$ is a binary erasure channel, the expression \eqref{eq:E(f)} can be simplified to obtain a tight bound. If $y$ is an erasure symbol, then $f(y) = 2$, and otherwise, $p(y) = 0$. This means that we simply have
 \[
   \E_{y_1, y_2 \sim p(y)} \sqrt{f(y_1)^2 + f(y_2)^2 - 4} = \E_{y \sim p(y)} f(y)
 \]
 and the equality follows. \qedhere
\end{proof}

\section{Rough polarization for erasure channels}
\label{s:binary-erasure-rough}
If $W$ is the binary erasure channel, we have $I(W_n) = 1 - Z(W_n)$ and $Z(W_n^-) = 2Z(W_n) - Z(W_n^2)$. In this case, we can show the following. 

\begin{proposition}
  For the binary erasure channel $W$, for all $\alpha \in (3/4, 1)$, there exists a constant $c_\alpha$ such that for all $\eps > 0$ and $m \ge c_\alpha \log(1/\epsilon)$ we have
  \[
    \Pr_i \left(Z(\wmi) \le 2 \alpha^m \right) \ge I(W) - \epsilon.
  \]
  \label{p:bec-rough}
\end{proposition}
\begin{proof}
  We can rearrange the evolution Equation \eqref{eq:evolution} and apply Propositions \ref{p:zplus} and \ref{p:zminus} for the BEC case to obtain the equation
\begin{equation}
  Y(W^{(i)}_{n+1}) = Y(\wln) \cdot
    \begin{cases}
      Z(\wln) (1 + Z(\wln)) & i \bmod 2 \equiv 1\\
      (1-Z(\wln))(2-Z(\wln)) & i \bmod 2 \equiv 0
    \end{cases}
  \label{eq:bec-y-evolution}
\end{equation}
Since \[ \sqrt{z(1+z)} + \sqrt{(1-z)(2-z)} \le \sqrt{3} \]
 for all $z \in [0, 1]$ (observed originally by \cite{arikan-telatar-arxiv}), and $Y(W) \le 1/4$, we can conclude the geometrically decaying upper bound $\E_i \Bigl[ \sqrt{Y(W^{(i)}_n)} \Bigr] \le \frac12 \left( \frac{\sqrt{3}}{2} \right)^{n}$. 
Therefore, by Markov's inequality, we have 
\begin{equation}
  \label{bec-markov}
  \Pr_i [Y(\wni) \ge \alpha^n ] \le \frac12 \left( \frac 3 {4 \alpha} \right)^{n/2}.
\end{equation}
We have $\E_i [Z(\wni)] = \E_i [Z(W_{n-1}^{(i)})] = Z(W)$, and so \[ \Pr_i(A_\alpha^b) \min_{i \in A_\alpha^b} Z(\wni) \le \E_i [\wni] = Z(W) \ . \] 
Since $A_\alpha^g \subset A_\alpha$ and $A_\alpha^g$ is disjoint from $A_\alpha^b$, we have $\Pr(A_\alpha^b) = 1 - \Pr(A_\alpha^g) - \Pr(\overline{A_\alpha})$, and we obtain
\begin{equation}
  \Pr(A_\alpha^g) \ge 
  1 - \frac{Z(W)}{\min_{i \in A_\alpha^b} (Z(\wni))} - \Pr(\overline{A_\alpha}) 
  \ge 
  1 - \frac{Z(W)}{1-2\alpha^n} - \frac12 \left( \frac 3 {4 \alpha} \right)^{n/2}
  \label{eq:bec-bound}
\end{equation}
where we have used \eqref{bec-markov} to bound the probability of $\overline{A_\alpha}$ and Fact \ref{fact:wni-approx} to lower bound $\min_{i \in A_\alpha^b} Z(W_n^{(i)})$.

By Fact \ref{fact:wni-approx}, $Z(\wni) \le 2 \alpha^n$ for $i \in A_\alpha^g$. Together with \eqref{eq:bec-bound} we can conclude that for all $\alpha \in (3/4, 1)$, there is some constant $c_\alpha$ such that for all $\eps > 0$ and $m \ge c_\alpha \log (1/\epsilon)$, so that
\[
  \Pr_i[Z(\wmi) \le 2 \alpha^{m}] \ge \Pr(A^g_\alpha) \ge 1 - Z(W) - \epsilon = I(W) - \epsilon. \qedhere
\]
\end{proof}
\section{Analytic Bound on \eqref{eq:h-change-bound}}
\label{s:h-change-bound}
Our objective is to show that the quantity
\begin{equation}
  \Upsilon(x) := \frac{h(2x(1-x)) - h(x)}{h(x)(1-h(x))} \ge \theta
  \label{eq:h-change-bound-reprint}
\end{equation}
for $x \in (0, 1/2)$, for some absolute constant $\theta > 0$.

We will establish this by splitting the analysis into three regimes: $x$ near $0$, $x$ near $1/2$, and $x$ bounded away from the boundaries $0$ and $1/2$. Specifically, we will consider the intervals $(0,a)$, $[a,1/2-a]$, and $(1/2-a,1/2)$ for some absolute constant $a$, $0 < a < 1/4$ (which will be specified in the analysis of the first case when $x$ is close to $0$):

\begin{description}
\item [Case 1: $x$ close to $0$.] Clearly, we have 
\[ \lim_{x \to 0} \frac{h(2x(1-x))}{h(x)} = 2 \ . \]
By continuity, there exists an $a$, $0 < a < 1/4$, such that $\frac{h(2x(1-x))}{h(x)} \ge \frac32$ for $0 < x < a$. This implies $\Upsilon(x) \ge \frac{1}{2 (1-h(x))} \ge \frac12$ for $x \in (0,a)$.

\item [Case 2: $a \le x \le 1/2-a$.] In this interval, we have
\begin{align*}
h(2x(1-x)) - h(x) &= h(x + 2x(1/2-x))-h(x) \\
&\ge h\bigl(x + 2 a (1/2-a)\bigr) - h(x) \quad \mbox{($h(x)$ is increasing on $(0,1/2)$)} \\
&\ge h\bigl(1/2-a + 2 a(1/2-a)\bigr) - h(1/2-a) \quad \mbox{(by concativy of $h(x)$)}
\end{align*}
Thus $\Upsilon(x) \ge  4 \Bigl( h\bigl(1/2-a + 2 a(1/2-a)\bigr) - h(1/2-a)\Bigr)$ and thus at least a positive constant depending on $a$, for $x \in [a,1/2-a]$

\item [Case 3: $x \in (1/2-a,1/2)$.]
For $x$ near $1/2$, we use the following inequality that is valid for $0 \le \gamma \le 1/2$
\begin{equation*}
  1 - 4 \gamma^2 \le h(1/2 - \gamma) \le 1 - 2 \gamma^2 \ .
\end{equation*}
The lower bound is implied by \eqref{eq:ZversusH} for $W = \bsc(p)$ and the upper bound, in fact with a better constant of $2/\ln 2$, follows from the Taylor expansion of $h(x)$ around $x = 1/2$.
Using the above, we have the following lower bound in the range $\xi \in (0,a)$:
\begin{align*}
\Upsilon(1/2-\xi) &= \frac{h(1/2 - 2\xi^2) - h(1/2-\xi)}{h(1/2-\xi)\cdot (1-h(1/2-\xi))} \\
& \ge \frac{(1- 4 (2\xi^2)^2) - (1-2\xi^2)}{ h(1/2-\xi) \cdot 4 \xi^2 } \\
& \ge \frac{2 \xi^2 - 16 \xi^4}{4 \xi^2}  \\
& \ge \frac12 - 4 a^2 \ .
\end{align*}
\end{description}
Thus in all three cases $\Upsilon(x)$ is lower bounded by an absolute positive constant, as desired.

\section{Output Symbol Binning}
\label{s:output-symbol-binning}

\begin{proof}[Proof of Proposition \ref{p:approx-z}]
  First, it is clear that the algorithm runs in time polynomial in $|\Y|$ and $k$; $k$ bits of precision is more than sufficient for all of the arithmetic operations, and the operations are done for each symbol in $\Y$.

For $\tilde y \in \widetilde \Y$, let $I_{\tilde y}$ be the set of $y$ associated with the symbol $\tilde y$; that is, all $y$ such that $p(0|y)$ falls in the interval of $[0,1]$ associated with $\tilde y$ (which is $[j/k, (j+1)/k)$ for $\tilde y = \tilde y_j$).

  For the lower bound, it is clear that $H(W) \le H(\widetilde W)$. Juxtaposing the definitions of $H(W)$ and $H(\widetilde{W})$ together, obtain (defining the binary entropy function $h(x) = -x \lg x - (1-x) \lg (1-x)$):
  \begin{align*}
  H(W)
  &= 
  \sum_{y \in Y} p(y) h(p(0|y))\\
  &\le 
  \sum_{\tilde y \in \tilde \Y} \left( \sum_{y \in I_{\tilde y}} p(y) \right)\left( h(p(0|\tilde y))) \right) 
  = H(\widetilde W)
\end{align*}
where the inequality is due to the concavity of $h(x)$.

Using the fact $\min_i a_i / b_i \le \sum_i a_i / \sum_i b_i \le \max_i a_i / b_i$, we can bound
\[
  p(0|\tilde y) = \frac{p(\tilde y|0)p(0)}{p(\tilde y)} = \frac12 \frac{\sum_{y \in I_{\tilde y}} p(y|0)}{\sum_{y \in I_{\tilde y}} p(y)}
\]
with the expressions
\[
  \min_{y \in I_{\tilde y}} p(0|y) \le p(0|\tilde y) \le \max_{y \in I_{\tilde y}} p(0 | y)
\]
which implies, for all $y \in I_{\tilde y}$,
\[
  p(0|\tilde y) - \frac1k \le p(0|y) \le p(0|\tilde y) + \frac1k.
\]

We will need to offer a bound on $h(p(0|\tilde y))$ as a function of $h(p(0|y))$.
$h(x)$ is concave and obeys $|h'(x)| \le \lg k$ if $1/k < x < 1 - 1/k$. Define the ``middle set'' $\tilde y_m = \{ \tilde y_i : 0 < i < k - 1\}$, corresponding with intervals where $p(0 | \tilde y_m)$ is in the range $1/k < x < 1-1/k$. Then, by the concavity of $h(x)$, for all $\tilde y \in \tilde y_m$ and $y \in I_{\tilde y}$, we have $h(p(0|\tilde y)) \le h(p(0|y)) + 2\lg(k)/k$.

We now provide a bound for the remaining symbols $\tilde y_0$, $\tilde y_{k-1}$ and $\tilde y_k$. $\tilde y_k$ is trivial because it represents all symbols where $p(0|y) = 1$, and merging those symbols together still results in $p(0|\tilde y) = 1$. For $\tilde y_0$, we have
\[
  h(p(0|\tilde y)) \le h(1/k) \le 2\lg(k)/k \le h(p(0|y)) + 2\lg(k)/k
\]
and similarly for $\tilde y_{k-1}$.

With these expressions in hand, we can now write
\begin{align*}
  H(\widetilde W) &= \sum_{\tilde y \in \Y} \sum_{y \in I_{\tilde y}} p(y)h(p(0|\tilde y)))\\
  &\le \sum_{\tilde y \in \Y} \sum_{y \in I_{\tilde y}} p(y)(h(p(0|y)) + 2\lg(k)/k)\\
  &\le H(W) + 2\lg(k)/k \ . \qedhere
\end{align*}
\end{proof}
\end{document}